\documentclass[11pt]{article}
\usepackage[utf8]{inputenc}
\usepackage{fullpage}
\usepackage{amsmath, amssymb, amsthm}
\usepackage[hidelinks]{hyperref}
\usepackage{thmtools}
\usepackage{nameref,cleveref}
\usepackage{color}
\usepackage[shortlabels]{enumitem}

\usepackage{svg}

\newtheorem{thm}{Theorem}
\newtheorem{proposition}[thm]{Proposition}
\newtheorem{lemma}[thm]{Lemma}
\newtheorem{theorem}[thm]{Theorem}
\newtheorem{corollary}[thm]{Corollary}
\newtheorem{definition}[thm]{Definition}

\crefname{definition}{definition}{definitions}
\Crefname{definition}{Definition}{Definitions}

\newcommand{\eps}{\varepsilon}
\newcommand{\heps}{\hat{\eps}}
\newcommand{\cX}{\mathcal{X}}
\newcommand{\cY}{\mathcal{Y}}
\newcommand{\R}{\mathbb{R}}
\newcommand{\Z}{\mathbb{Z}}

\newcommand{\bbP}{\mathbb{P}}
\newcommand{\dr}[2]{\mathrm{D}_{\alpha}\left(#1\ \middle\|\ #2\right)}
\newcommand{\dra}[3]{\mathrm{D}_{#3}\left(#1\ \middle\|\ #2\right)}

\DeclareMathOperator{\DLap}{DLap}
\DeclareMathOperator{\Lap}{Lap}
\DeclareMathOperator{\csch}{csch}
\DeclareMathOperator{\hot}{hot}
\DeclareMathOperator{\supp}{supp}

\allowdisplaybreaks

\title{Exact zCDP Characterizations for\\ Fundamental Differentially Private Mechanisms}

\author{Charlie Harrison \\Google\\ \texttt{\small csharrison@google.com} \and Pasin Manurangsi \\Google Research\\ \texttt{\small pasin@google.com}}

\begin{document}

\maketitle

\begin{abstract}
Zero-concentrated differential privacy (zCDP) is a variant of differential privacy (DP) that is widely used partly thanks to its nice composition property. 
While a tight conversion from $\eps$-DP to zCDP exists for the \emph{worst-case} mechanism, many common algorithms satisfy stronger guarantees. In this work, we derive tight zCDP characterizations for several fundamental mechanisms. We prove that the tight zCDP bound for the $\eps$-DP Laplace mechanism is exactly $\eps + e^{-\eps} - 1$, confirming a recent conjecture by Wang~\cite{tight-laplace-tweet}. We further provide tight bounds for the discrete Laplace mechanism, $k$-Randomized Response (for $k \leq 6$), and RAPPOR. Lastly, we also provide a tight zCDP bound for the worst case bounded range mechanism.
\end{abstract}

\renewcommand{\arraystretch}{1.5}
\begin{center}
\begin{table}
\begin{center}
\begin{tabular}{ |c|c|c|c|c| }
 \hline
Mechanism & Tight zCDP bound & $\eps \to 0$ & Reference \\ 
 \hline
 $\eps$-DP & $\eps \tanh(\eps/2)$ 
 & $\eps^2 / 2$
 & \cite{DPorg-pdp-to-zcdp} \\ 
 \hline
 $\eps$-DP Laplace & $\eps + e^{-\eps} - 1$ 
 & $\eps^2 / 2$
 & \Cref{thm:zcdp-lap}\\ 
 \hline
 $\eps$-DP discrete Laplace & $\eps\left(1 - \frac{1}{\Delta} (1 - e^{-\eps}) \csch(\eps / \Delta)\right)$ 
& $\eps^2 / 2$
 & \Cref{thm:zcdp}
 \\ 
 \hline
 $\eps$-DP $k$-RR & $ \frac{\eps(e^{\eps} - 1)}{e^{\eps} - 1 + k}$ for $k \le 6$
 & $\eps^2 / k$ for $k \le 6$
 & \Cref{thm:zcdp-rr-tight}
 \\ 
 \hline
 $\eps$-DP RAPPOR & $\eps \tanh(\eps / 4)$ 
 & $\eps^2 / 4$
 & \Cref{prop:zcdp-rappor}
 \\ 
 \hline
 $\eps$-Bounded Range & $\frac{\eps}{e^\eps - 1} + \log\left(
\frac{e^{\eps }-1}{\eps }\right) - 1$ 
& $\eps^2 / 8$
& \Cref{thm:br-cdp}
\\ 
 \hline
\end{tabular}
\end{center}
\caption{Our main contributions: tight zCDP bounds for various $\eps$-DP mechanisms, as well as all $\eps$-bounded range mechanisms. Also included are the asymptotics for all tight bounds as $\eps \to 0$, which is common in practice in the large-composition regime.}
\label{tab:main}
\end{table}
\end{center}

\section{Introduction}
Differential privacy (DP) \cite{dwork-calibrating} is a strong, formal notion of privacy which bounds the information revealed by the output of an algorithm. In recent years, various \emph{relaxations} of $\eps$-differential privacy (aka \emph{pure}-DP) have emerged. Most relevant to this work are R{\'e}nyi differential privacy (RDP) \cite{Mironov17} and zero-concentrated differential privacy (zCDP) \cite{DworkR16,bun-steinke-16}, which have proven useful in regimes where many algorithms are composed together.

Obtaining tight privacy guarantees is of paramount importance in both theoretical and practical applications of differential privacy, so when the algorithms being composed in this manner are themselves $\eps$-DP, tight translations between $\eps$-DP and zCDP is critical. Bun and Steinke~\cite{bun-steinke-16} showed that any $\eps$-DP mechanism also satisfies $(\eps^2/2)$-zCDP, and it is easy to show that any $\eps$-DP mechanism satisfies $\eps$-zCDP. Noticeably, there is a gap between these two bounds, as the former is better in the low-$\eps$ regime but the latter is better in the large-$\eps$ regime. This gap was closed in \cite{DPorg-pdp-to-zcdp} which showed that any $\eps$-DP mechanism tightly satisfies $\left(\eps \tanh(\eps/2)\right)$-zCDP and the binary Randomized Response (RR) is the worst-case mechanism. Nevertheless, many fundamental mechanisms enjoy better zCDP guarantees than such a generic formula implies.
The main contribution of our work is in proving the tight zCDP guarantees of several fundamental $\eps$-DP mechanisms, including the Laplace mechanism, discrete Laplace mechanism, RAPPOR and $k$-RR (when $k$ is sufficiently small). The bounds are listed in \Cref{tab:main}. In the case of the Laplace mechanism, we confirm a bound conjectured in \cite{tight-laplace-tweet}, and formalize the RAPPOR bound first stated in \cite{harrison-rappor}. The bounds for the other mechanisms are novel to the best of our understanding.

In addition to the aforementioned mechanisms, we also provide tight zCDP bounds for \emph{$\eps$-bounded range ($\eps$-BR)} mechanisms. The notion of bounded range mechanisms was proposed by
\cite{durfee19-pay-what-you-get,dong-optimal-expo-composition} in an attempt to achieve better composition properties for the exponential mechanism \cite{msherry-exponential}. From this foundation, \cite{cesar-bounding-concentrating-br} proved that the exponential mechanism satisfies $(\eps^2 / 8)$-zCDP, which is only tight asymptotically as $\eps \to 0$. We show a tight bound for worst-case $\eps$-bounded range mechanisms (which includes the exponential mechanism) in \Cref{thm:br-cdp}.

\paragraph{Properties of our bounds.} For all of our tight zCDP bounds, we find that they are exactly the same as the KL divergence between output distributions on neighboring datasets.
That is, the RDP bound is achieved when the R{\'e}nyi divergence order approaches one. However, proving such a statement proved challenging and, for each mechanism, we proceed with different parameterization (and differentiation) to prove such an inequality. Another interesting--and arguably surprising--finding is that the zCDP bound for $k$-RR does \emph{not} occur at $\alpha \to 1$ for any sufficiently large $k$ (\Cref{lem:non-opt-rr}). In this case, we give looser non-tight bounds.
\section{Preliminaries}

\subsection{Privacy and Related Tools}

\paragraph{Privacy notions.} We recall the various notions of differential privacy considered in this work.
\begin{definition}[Differential privacy \cite{dwork-calibrating}] \label{def:dp}
A randomized algorithm $M: \mathcal{X}^n \to \mathcal{Y}$ satisfies \emph{$\eps$-differential privacy ($\eps$-DP)} if, for all neighboring\footnote{For the purpose of our work, the exact definition of neighboring dataset is unimportant. Nevertheless, one may think of $x, x'$ as neighbors iff they differ on a single row for simplicity.} inputs $x, x' \in \mathcal{X}^n$ and for all $S \subseteq \mathcal{Y}$, $
\log \left(
\frac{\mathbb{P}(M(x) \in S)}{\mathbb{P}(M(x') \in S)}\right) \le \eps
$.
\end{definition}

\begin{definition}[R{\'e}nyi divergence \cite{renyi61}]
Let $P$ and $Q$ be probability distributions on $\Omega$. Then the R{\'e}nyi divergence between $P$ and $Q$ at order $\alpha$, denoted $\dr{P}{Q}$, is 

$$
\dr{P}{Q} = \frac{1}{\alpha-1} \log\left(\int_\Omega P(x)^\alpha Q(x)^{1-\alpha} \mathrm{d}x\right),
$$

where $P(\cdot), Q(\cdot)$ denote the probability mass/density functions of $P$ and $Q$, respectively.\footnote{In this work, $P$ and $Q$ are always absolutely continuous with respect to one another.} The R{\'e}nyi divergence at $\alpha = 1$ is defined as 

$$
\dra{P}{Q}{1} = \lim_{\alpha \to 1} \dr{P}{Q} = D_{KL}(P || Q) = \int_{\Omega} P(x) \log \left(\frac{P(x)}{Q(x)} \right)\mathrm{d}x.
$$
\end{definition}

For notational convenience, we will sometimes write random variables in place of their distributions in $\dr{\cdot}{\cdot}$ when there is no ambiguity.

\begin{definition}[R{\'e}nyi DP \cite{Mironov17}] \label{def:rdp}
A randomized algorithm $M: \mathcal{X}^n \to \mathcal{Y}$ satisfies \emph{$(\alpha, \hat{\eps})$-R{\'e}nyi differential privacy ($(\alpha, \hat{\eps})$-RDP)} if, for all neighboring inputs $x, x' \in \mathcal{X}^n$,
$
\dr{M(x)}{M(x')} \le \hat{\eps}.
$

Furthermore, we say that $M$ is \emph{tightly-$(\alpha, \hat{\eps})$-RDP} if it is $(\alpha, \hat{\eps})$-RDP and it is not $(\alpha, \hat{\eps}')$-RDP for any $\hat{\eps}' < \hat{\eps}$.
\end{definition}

\begin{definition}[Concentrated DP \cite{bun-steinke-16}]
A randomized algorithm $M: \mathcal{X}^n \to \mathcal{Y}$ satisfies \emph{$\rho$-zero concentrated differential privacy ($\rho$-zCDP)} if it satisfies $(\alpha, \rho\alpha)$-RDP for all $\alpha \geq 1$.

Furthermore, we say that $M$ is \emph{tightly-$\rho$-zCDP} if it is $\rho$-zCDP and it is not $\rho'$-zCDP for any $\rho' < \rho$.
\end{definition}

Finally, we recall the notion of bounded range mechanisms.

\begin{definition}[Bounded range DP \cite{durfee19-pay-what-you-get, dong-optimal-expo-composition, DPorg-exponential-mechanism-bounded-range}] \label{def:br}
A randomized algorithm $M: \mathcal{X}^n \to \mathcal{Y}$ satisfies \emph{$\eta$-bounded range ($\eta$-BR)} if, for all neighboring inputs $x, x' \in \mathcal{X}^n$ differing on a single row and for all $S \subseteq \mathcal{Y}$, there exists a $t \in \mathbb{R}$ such that
$
\log \left(
\frac{\mathbb{P}(M(x) \in S)}{\mathbb{P}(M(x') \in S)}\right) \in [t, t + \eta].
$
\end{definition}
It is immediate from the definition that $\eta$-BR implies $\eta$-DP, and $\eps$-DP implies $(2 \eps)$-BR.

\paragraph{Sensitivity.} We will consider mechanisms that simply adds noise to the output of a function. The privacy guarantees of this mechanism depends on the sensitivity of the function, defined as follows.
\begin{definition}[Sensitivity]
The sensitivity of a function $g:  \mathcal{X}^n \to \mathbb{R}$ is $\Delta(g) := \max_{X, X'} |g(X) - g(X')|$ where the maximum is over all neighboring inputs $X$ and $X'$. 
\end{definition}

\paragraph{Noise distributions.}
\begin{itemize}
\item The continuous Laplace distribution with scale parameter $\lambda$ is denoted  by $\Lap(\lambda)$ and has support on $\mathbb{R}$. Its probability density function is
$
f_{\Lap(\lambda)}(x) = \frac{1}{2 \lambda} e^{-|x| / \lambda}.
$

\item The discrete Laplace distribution with scale parameter $a$ is denoted by $\DLap(a)$ and has support on $\mathbb{Z}$. Its probability mass function is 
$
f_{\DLap(a)}(x) = \tanh(a/2) e^{-a|x|}.
$
\end{itemize}
\paragraph{Existing tight bounds.} We will leverage (and extend) the following existing tight bounds.

\begin{proposition}[DP to RDP \cite{DPorg-pdp-to-zcdp}]\label{prop:dp-to-rdp}
Let $M: \mathcal{X}^n \to \mathcal{Y}$ satisfy $\eps$-DP, then
$M$ satisfies $(\alpha, \hat{\eps}(\alpha))$ RDP for all $\alpha > 1$ where
$\hat{\eps}(\alpha) = \frac{1}{\alpha -1}\log\left(
\frac{e^{\alpha \eps} + e^{\eps(1-\alpha)}}{e^{\eps} + 1}
\right)$.
\end{proposition}

\begin{proposition}[DP to zCDP \cite{DPorg-pdp-to-zcdp}]\label{prop:dp-to-zcdp}
Let $M: \mathcal{X}^n \to \mathcal{Y}$ satisfy $\eps$-DP, then
$M$ satisfies $\eps \cdot \tanh(\eps / 2)$-zCDP.
\end{proposition}

Along the way to proving the above conversions, Steinke also proved that the worst case occurs at $\alpha \to 1$, which will state explicitly below for convenience.

\begin{proposition} \label{cor:dlap-rdp-one}
Let $M: \mathcal{X}^n \to \mathcal{Y}$ satisfy $\eps$-DP. Then for every $\alpha > 1, \eps > 0$, $M$ satisfies $(\alpha, \hat{\eps}(\alpha))$ RDP where $\eps \cdot \tanh(\eps/2) - \frac{\hat{\eps}(\alpha)}{\alpha} \ge 0$.
\end{proposition}

\subsection{Some Useful Functions and Their Properties}

We recall the definitions of the hyperbolic functions. The hyperbolic sine and cosine are defined for any $x \in \mathbb{R}$ by
\[ \sinh(x) = \frac{e^x - e^{-x}}{2} \quad \text{and} \quad \cosh(x) = \frac{e^x + e^{-x}}{2}. \]
From these, the other four hyperbolic functions are derived:
\[ \tanh(x) = \frac{\sinh(x)}{\cosh(x)}, \quad \coth(x) = \frac{\cosh(x)}{\sinh(x)}, 
\quad \text{and} \quad \csch(x) = \frac{1}{\sinh(x)}. \]
These functions are related by the fundamental identity $\cosh^2(x) - \sinh^2(x) = 1$. We also note the parity of the basic functions: $\cosh(x)$ is an even function, satisfying $\cosh(-x) = \cosh(x)$, while $\sinh(x)$ is an odd function, satisfying $\sinh(-x) = -\sinh(x)$.

The function $\frac{\sinh(x)}{x}$ will appear often in the proofs, so we list a few of its properties below. Throughout, we let $\frac{\sinh 0}{0} = \lim_{x \to 0} \frac{\sinh(x)}{x} = 1$ for convenience.

\begin{lemma} \label{lem:sinhx-divx-inc}
$\frac{\sinh x}{x}$ is increasing for $x \in [0, \infty)$.
\end{lemma}

\begin{proof}
We have $\frac{d \frac{\sinh x}{x}}{dx} = \frac{x \cosh x - \sinh x}{x^2} = \frac{\sinh x \left(\coth x - \frac{1}{x}\right)}{x}$. It is well known that $\tanh(x) < x$ for all $x > 0$. (See e.g. \cite{tanhx-greater-than-x}.) Thus, $x, \sinh x$ and $\coth(x) - 1/x$ are all non-negative for $x > 0$.
\end{proof}

\begin{lemma} \label{lem:log-sinh-over-x-convex-and-superadditive}
The function $\phi(x) = \log\left(\frac{\sinh x}{x}\right)$ is strictly convex and superadditive for $x \in [0, \infty)$.
\end{lemma}
\begin{proof}
First, we compute the derivatives of $\phi(x) = \log(\sinh x) - \log x$.
\begin{align*}
\phi'(x) &= \frac{\cosh x}{\sinh x} - \frac{1}{x} = \coth x - \frac{1}{x} \\
\phi''(x) &= -\csch^2 x + \frac{1}{x^2} = \frac{1}{x^2} - \frac{1}{\sinh^2 x}
\end{align*}
Since $\sinh x \geq x$ for all $x > 0$ (which follows from \Cref{lem:sinhx-divx-inc}), we have $\phi''(x)$.

To prove superadditivity, since $\phi$ is convex, $\phi'$ is increasing. For $y\geq 0$, let $H(x) = \phi(x+y) - \phi(x)$. Then, $H'(x) = \phi'(x+y) - \phi'(x) \geq 0$, so $H(x)$ is increasing.
Thus, for any $x \geq 0$, $H(x) \geq H(0)$, which implies that $\phi(x + y) \geq \phi(x) + \phi(y) - \phi(0) = \phi(x) + \phi(y)$ as desired.
\end{proof}
What follows is a simple lemma which allows us to leverage concavity of the R{\'e}nyi curve.

\begin{lemma}\label{lem:concavity}
    Let $f(x)$ be concave and differentiable on the domain $[1, \infty)$. If $f'(1) \le f(1)$, then $g(x) = f(x) / x$ is maximized at $x = 1$.
\end{lemma}
\begin{proof}
    It suffices to show $g'(x) = \frac{x f'(x) - f(x)}{x^2} \le 0\iff h(x) = x f'(x) - f(x) \le 0$. This holds, as $h(1) = f'(1) - f(1) \le 0$ by construction, and $h'(x) = x f''(x) \le 0$ by concavity.
\end{proof}

\section{Central DP Mechanisms}

In this section, we derive the tight zCDP bounds for  Laplace and discrete Laplace mechanisms. In both cases, we show that the tight zCDP bound occurs when $\alpha \to 1^+$.

\subsection{Laplace Mechanism}

The Laplace mechanism was the first mechanism proposed alongside the definition of differential privacy in \cite{dwork-calibrating}. For any function $g: \cX^n \to \R$, it simply works by outputting $g(x) + Z$ where $Z \sim \Lap(\Delta(g) / \eps)$. This mechanism achieves $\eps$-DP. 

In the work proposing RDP~\cite{Mironov17}, Mironov showed that the Laplace mechanism is tightly-$(\alpha, \hat{\eps}(\alpha))$-RDP for $\alpha > 1$ where  
$$
\hat{\eps}_{L}(\alpha) = 
\frac{1}{\alpha-1}\log\left(
\frac{\alpha}{2 \alpha - 1}e^{(\alpha - 1) / \lambda}
+ \frac{\alpha - 1}{2 \alpha - 1}e^{-\alpha / \lambda}
\right)
$$

and $\lambda = \frac{\Delta(g)}{\eps}$. Given the above bound, it has been conjectured \cite{tight-laplace-tweet} that the tight zCDP bound occurs when $\alpha \to 1$, which yields $\lim_{\alpha \to 1} \heps(\alpha) = \eps + e^{-\eps} - 1$. However, the proof of such a statement has never been published. Below we give a proof of this conjecture.

\begin{theorem}[zCDP for Laplace] \label{thm:zcdp-lap}
The $\eps$-DP Laplace mechanism is tightly-$(\eps + e^{-\eps} - 1)$-zCDP.
\end{theorem}

\begin{proof}
We start by showing that the mechanism is $(\eps + e^{-\eps} - 1)$-zCDP. From the aforementioned RDP bound, our goal is to prove that
\[
f(\eps) = (\eps + e^{-\eps} - 1) - \frac{1}{\alpha} \cdot \hat{\eps}_{L}(\alpha) = (\eps + e^{-\eps} - 1) - \frac{1}{\alpha(\alpha-1)} \log\left( \frac{\alpha}{2\alpha-1} e^{(\alpha-1)\eps} + \frac{\alpha-1}{2\alpha-1} e^{-\alpha \eps} \right)
\]
is non-negative for all $\eps > 0$ and $\alpha > 1$. To do this, it suffices to argue that $f(0) = 0$ and $f'(\eps) \ge 0$ for $x > 0$. Simple calculation confirms that $f(0) = 0$, and that
\begin{align*}
f'(\eps) &= (1 - e^{-\eps}) - \frac{e^{(\alpha-1)\eps} - e^{-\alpha \eps}}{\alpha e^{(\alpha-1)\eps} + (\alpha-1) e^{-\alpha \eps}} 
&= \frac{2 e^{\alpha  \eps }\alpha(\alpha - 1)\eps \left(\frac{\sinh(\alpha \eps)}{\alpha \eps} - \frac{\sinh((\alpha - 1)\eps)}{(\alpha - 1)\eps}\right)}{(\alpha -1) e^{\eps }+\alpha  e^{2 \alpha  \eps}} &\geq 0, 
\end{align*}
where the last inequality is because all terms are negative; in particular, the last term $\frac{\sinh(\alpha \eps)}{\alpha \eps} - \frac{\sinh((\alpha - 1)\eps)}{(\alpha - 1)\eps}$ is non-negative due to \Cref{lem:sinhx-divx-inc}. Finally, the bound is tight as $\lim_{\alpha \to 1} \heps(\alpha) = \eps + e^{-\eps} - 1$.
\end{proof}
\subsection{Discrete Laplace}

When the range of $g$ is integers, the Laplace mechanism can be improved by using the discrete Laplace (aka Geometric) mechanism~\cite{GhoshRS12}. In particular, for $g: \cX^n \to \Z$, the $\eps$-DP discrete Laplace mechanism outputs $g(x) + Z$ where $Z \sim \DLap(\eps/\Delta(g))$. The privacy profile of this mechanism depends on the value of the sensitivity, $\Delta(g)$. Thus, our bound depends on $\Delta := \Delta(g)$ as well. We state its tight bound for zCDP below.

\begin{theorem}[zCDP of discrete Laplace]\label{thm:zcdp}
The $\eps$-DP discrete Laplace mechanism for sensitivity $\Delta$ is tightly-$\left(\eps\left(1 - \frac{(1 - e^{-\eps}) \csch(\eps / \Delta)}{\Delta}\right)\right)$-zCDP.
\end{theorem}

Before we prove the above, let us first remark that, when $\Delta = 1$, \Cref{thm:zcdp} yields the bound of $\eps \cdot \tanh(\eps/2)$, matching the worst case bound for $\eps$-DP algorithm in \Cref{prop:dp-to-zcdp}. This is as expected, since the privacy loss distribution\footnote{See e.g. \cite{PLD-supp} for the definition of privacy loss distributions (PLDs) and derivations of PLDs for discrete Laplace and Randomized Response.} of the $\eps$-DP discrete Laplace mechanism for $\Delta = 1$ coincides with the $\eps$-DP binary Randomized Response, which is known to be the worst case $\eps$-DP mechanism in terms of its zCDP guarantee.

On the other hand, if we take $\Delta \to \infty$, then we arrive at $\lim_{\Delta \to \infty} \left(\eps\left(1 - \frac{(1 - e^{-\eps}) \csch(\eps / \Delta)}{\Delta}\right)\right) = \eps + e^{-\eps} - 1$, which coincides with the zCDP bound of the Laplace mechanism (\Cref{thm:zcdp-lap}). Again, this is expected since the behavior of discrete Laplace mechanism more closely resembles the Laplace mechanism as $\Delta \to \infty$.

\subsubsection{RDP Bound Derivation}

We will next prove \Cref{lem:dlap-rdp}. We start by directly computing the R{\'e}nyi divergence between the discrete Laplace distribution and its shift, as stated below.

\begin{proposition}
\label{prop:dlap-rdp}
Let $Z \sim \DLap(a)$ and $d \in \Z_{\geq 0}$, then
\begin{align*}
&\dr{Z + d}{Z} = \dr{Z}{Z + d} \\ &\qquad= \frac{1}{\alpha-1} \log\left(
\tanh(a/2) \left(
    \frac{e^{-a \alpha d}}{e^{a}-1}
    +
    \frac{e^{a-a \alpha d }-e^{a (\alpha  (d +2)-d )}}{e^a-e^{2 a \alpha }}
    +
    \frac{e^{-a (1-\alpha) d}}{e^{a}-1}
\right)
\right)
\end{align*}
\end{proposition}
\begin{proof}
We can compute $\dr{Z + d}{Z}$ directly as follows.
\begin{align*}
&e^{(\alpha - 1)\dr{Z + d}{Z}} = 
    \sum_{x=-\infty}^\infty f_{\DLap(a)}(x-d)^\alpha \cdot f_{\DLap(a)}(x)^{1-\alpha}\\
&= \sum_{x=-\infty}^\infty (\tanh(a/2)e^{-a|x-d|})^\alpha \cdot (\tanh(a/2)e^{-a|x|})^{1-\alpha}\\
&= \tanh(a/2)\sum_{x=-\infty}^\infty
    e^{-a (\alpha |x-d| + (1-\alpha)|x|)}\\
&= \tanh(a/2) \left(
    \left(\sum_{x=-\infty}^{-1}
    e^{a(x-\alpha d)}
    \right) +
        \left(\sum_{x=0}^{d}
        e^{a ((2 \alpha -1) x-\alpha d )}
    \right)+
        \left(\sum_{x=d+1}^{\infty}
        e^{a(\alpha d- x)}
    \right)
\right)\\
&= \tanh(a/2) \left(
    \frac{e^{-a \alpha d}}{e^{a}-1}
    +
    \frac{e^{a-a \alpha  d }-e^{a (\alpha  (d +2)-d )}}{e^a-e^{2 a \alpha }}
    +
    \frac{e^{-a (1-\alpha) d}}{e^{a}-1}
\right).
\end{align*}
Finally, $\dr{Z + d}{Z} = \dr{Z}{Z + d}$ follows from the symmetry of $\DLap$ around zero.
\end{proof}

We can now prove the RDP bound 
by showing that the worst case divergence occurs when $d = \Delta$.

\begin{lemma}[RDP of discrete Laplace]\label{lem:dlap-rdp}
The $\eps$-DP discrete Laplace mechanism for sensitivity $\Delta$ is tightly-$(\alpha, \heps_{D}(\alpha; \Delta))$-RDP for all $\alpha > 1$ where
\begin{align*}
\heps_{D}(\alpha; \Delta) = \frac{1}{\alpha-1} \log\left(
\tanh(a/2) \left(
    \frac{e^{-a \alpha \Delta}}{e^{a}-1}
    +
    \frac{e^{a-a \alpha  \Delta }-e^{a (\alpha  (\Delta +2)-\Delta )}}{e^a-e^{2 a \alpha }}
    +
    \frac{e^{-a (1-\alpha) \Delta}}{e^{a}-1}
\right)
\right)
\end{align*}
with $a = \eps / \Delta$.
\end{lemma}

\begin{proof} 
Let $f(d) = \frac{e^{-a \alpha d}}{e^{a}-1}+\frac{e^{a-a \alpha  d }-e^{a (\alpha  (d +2)-d )}}{e^a-e^{2 a \alpha }}+ \frac{e^{-a (1-\alpha) d}}{e^{a}-1}$. Notice that we have
\begin{align*}
f'(d) &= \frac{-a\alpha e^{-a \alpha d}}{e^{a}-1}+\frac{-a \alpha e^{a-a \alpha  d }- a(\alpha - 1) e^{a (\alpha  (d +2)-d )}}{e^a-e^{2 a \alpha }}+ \frac{a\alpha e^{-a (1-\alpha) d}}{e^{a}-1} \\
&= \frac{a\alpha(e^{a(\alpha - 1)d} - e^{-a\alpha d})}{e^a - 1} + \frac{a \alpha e^{a-a \alpha  d } + a(\alpha - 1) e^{a (\alpha  (d +2)-d )}}{e^{2 a \alpha } - e^a},
\end{align*}
which is term-by-term non-negative. Thus, $f$ is an increasing function in $d$.

Let $Z \sim \DLap(a)$.
For any function $g$ and neighboring inputs $x, x'$, we have
\begin{align*}
\dr{g(x) + Z}{g(x') + Z} = \dr{|g(x) - g(x')| + Z}{Z} \leq \dr{\Delta + Z}{Z} = \heps_{D}(\alpha; \Delta),
\end{align*}
where the inequality holds because $f$ is increasing. Thus, the mechanism is $(\alpha, \heps_D(\alpha; \Delta))$-RDP.

The tightness follows from considering any pair of inputs $x, x'$ such that $|g(x) - g(x')| = \Delta$.
\end{proof}

\subsubsection{From RDP to zCDP}

\begin{proof}[Proof of \Cref{thm:zcdp}]
Consider the following as a function of $\Delta$,
\begin{align*}
f(\Delta) =  a (\Delta + \left(e^{-a \Delta }-1\right) \csch(a))
- \heps_{D}(\alpha; \Delta) /\alpha.
\end{align*}

The upper bound follows from showing that $f(\Delta) \ge 0$ for all $\Delta \ge 1, a > 0, \alpha > 1$. Recall $\heps$ from \Cref{prop:dp-to-rdp}.
Observe that $f(1) = \eps \cdot \tanh(\eps/2) - \frac{\hat{\eps}(\alpha)}{\alpha}$, which is non-negative due to \Cref{cor:dlap-rdp-one}. Thus, to show that $f(\Delta) \geq 0$ for all $\Delta \geq 1$, it suffices to show that $f'(\Delta)$ is non-negative for all $\Delta \geq 1$.

It will be helpful for us to write out the first and second derivative of $f$ with respect to $\Delta$:
\begin{align}
    f'(\Delta) &=
    \frac{a \left(\frac{(2 \alpha -1) \left(e^{2 a \alpha }-1\right) e^{2 \alpha  a \Delta +a}}{\alpha  \left(-e^{a (2 \alpha +\Delta )}+e^{2 \alpha  a \Delta +a}-e^{2 \alpha  a (\Delta +1)+a}+e^{a (\Delta +2)}\right)}-a (\alpha -1) e^{-a \Delta } \csch(a)+\alpha \right)}{\alpha -1},\nonumber \\
    f''(\Delta) &= a^2 \left(-\frac{(1-2 \alpha )^2 \left(e^{2 a \alpha } - e^{2 a}\right) \left(e^{2 a \alpha }-1\right) e^{a \Delta +a}}{(\alpha -1) \alpha  \left(e^{a (-\alpha  \Delta +\Delta +2)}-e^{a (\alpha  (-\Delta )+2 \alpha +\Delta )}+e^{\alpha  a \Delta +a}-e^{\alpha  a (\Delta +2)+a}\right)^2}+a e^{-a \Delta } \csch(a)\right). \label{eq:second-diff-f}
\end{align}

To show that $f'(\Delta)$ is always non-negative for $\Delta \ge 1$, we will first show that $f'(1) \ge 0$, and then show $f''(\Delta) \ge 0$ for $\Delta \ge 1$.

\paragraph{Proof of Non-negativity of $f'(1)$.}
Note that
$$
f'(1) = a \left(-\frac{1}{\alpha} \cdot \frac{(\alpha -1) e^{4 a \alpha }+e^{2 a} \alpha +(1-2 \alpha ) e^{2 a \alpha }}{(\alpha -1)  \left(e^{4 a \alpha } - e^{2 a}\right)}+a -a \coth (a) +1\right).
$$
It is simple to check that $\lim_{\alpha \to 1^+}f'(1) = 0$. Thus, it suffices for us to show that a $f'(1)$ is increasing in $\alpha$. Since $\frac{1}{\alpha}$ is decreasing in $\alpha$ and $a -a \coth (a) +1$ is a constant (w.r.t. to $\alpha$), it in turn suffices for us to prove that  $\frac{(\alpha -1) e^{4 a \alpha }+e^{2 a} \alpha +(1-2 \alpha ) e^{2 a \alpha }}{(\alpha -1)  \left(e^{4 a \alpha } - e^{2 a}\right)}$ is decreasing in $\alpha$. We can rearrange this terms as follows:
\begin{align*}
\frac{(\alpha -1) e^{4 a \alpha }+e^{2 a} \alpha +(1-2 \alpha ) e^{2 a \alpha }}{(\alpha -1)  \left(e^{4 a \alpha } - e^{2 a}\right)} &= 1-
\frac{(2 \alpha -1) \left(e^{2 a \alpha } - e^{2 a}\right)}{(\alpha -1)  \left(e^{4 a \alpha } - e^{2 a}\right)} \\ &= 1 - e^{-a\alpha} \cdot \frac{(2 \alpha - 1)\sinh (a \cdot \alpha -a )}{(\alpha -1)\sinh (2 a \cdot \alpha - a )} \\
&=  1 - e^{-a \alpha + \phi(a(\alpha - 1)) - \phi(a(2\alpha - 1))},
\end{align*}
where $\phi$ is defined in \Cref{lem:log-sinh-over-x-convex-and-superadditive} as $\phi(x) = \log\left(\frac{\sinh(x)}{x}\right)$.

Thus, it suffices to show that $-a \alpha + \phi(a(\alpha - 1)) - \phi(a(2\alpha - 1))$ is decreasing in $\alpha$. To see this, notice that its derivative (in $\alpha$) is
\begin{align*}
-a + a \phi'(a(\alpha - 1)) - 2a\phi'(a(2\alpha - 1)) \leq -a - a\phi'(a(2\alpha - 1)) \leq 0,
\end{align*}
where the first inequality follows from the convexity of $\phi$ (\Cref{lem:log-sinh-over-x-convex-and-superadditive}) and the second is from $\phi$ is increasing (\Cref{lem:sinhx-divx-inc}). Hence, we have concluded our proof that $f'(1) \geq 0$ (for all $\alpha > 1, a > 0$).

\paragraph{Proof of Non-negativity of $f''(\Delta)$.}
To show non-negativity of $f''(\Delta)$, we will compare the log ratio of the two terms within the parenthesis in \eqref{eq:second-diff-f}. Denote this quantity by

$$
p(\Delta) = \log \frac
{a e^{-a \Delta } \csch(a)}
{-\frac{(1-2 \alpha )^2 \left(e^{2 a \alpha } - e^{2 a}\right) \left(e^{2 a \alpha }-1\right) e^{a \Delta +a}}{(\alpha -1) \alpha  \left(e^{a (-\alpha  \Delta +\Delta +2)}-e^{a (\alpha  (-\Delta )+2 \alpha +\Delta )}+e^{\alpha  a \Delta +a}-e^{\alpha  a (\Delta +2)+a}\right)^2}}.
$$

To prove $f''(\Delta) \ge 0$, it suffices to prove $p(\Delta) \geq 0$ for all $\Delta \geq 1$. For the latter, it in turn suffices to prove that $p(1) \ge 0$ and $p''(\Delta) \ge 0$ for all $\Delta \geq 1$.
To do this, let us first write out the relevant expressions:
\begin{align*}
p(1) &= \log \left(\frac{a (\alpha -1) \alpha  \csch(a) \sinh ^2(a-2 a \alpha ) \csch(a \alpha ) \csch(a \alpha  - a)}{(1-2 \alpha )^2}\right)\\
p'(\Delta) &= \frac{2 a (2 \alpha -1) \left(e^{2 a \alpha }-1\right) e^{2 \alpha  a \Delta +a}}{e^{a (2 \alpha +\Delta )}-e^{2 \alpha  a \Delta +a}+e^{2 \alpha  a (\Delta +1)+a}-e^{a (\Delta +2)}}-2 a \alpha\\
\end{align*}

We can rewrite $p(1)$ as follows:
\begin{align*}
p(1) 
&= \log \left(\frac{a (\alpha -1) \alpha  \csch(a) \sinh ^2(a-2 a \alpha ) \csch(a \alpha ) \csch(a \alpha  - a)}{(1-2 \alpha )^2}\right) \\
&= 2\phi((2\alpha - 1)a) - \phi((\alpha - 1)a) - \phi(a\alpha) - \phi(a) \\
&= \left(\phi((2\alpha - 1)a)  - \phi((\alpha - 1)a) - \phi(a\alpha)\right) + \left(\phi((2\alpha - 1)a) - \phi(a)\right) &\geq 0,
\end{align*}
where the inequality is from the fact that $\phi$ is super-additive (\Cref{lem:log-sinh-over-x-convex-and-superadditive}) and increasing (\Cref{lem:sinhx-divx-inc}).

As for $p'(\Delta)$, we have
\begin{align*}
p'(\Delta) &= \frac{2 a (2 \alpha -1) \left(e^{2 a \alpha }-1\right) e^{2 \alpha  a \Delta +a}}{e^{a\Delta}(e^{2a\alpha} - e^{2a}) + e^{a(2\alpha\Delta+1)}(e^{2a\alpha} - 1)}-2 a \alpha
\end{align*}
Dividing this expression by $2a$ and rearranging, we have
\begin{align*}
p'(\Delta) \geq 0 &\iff (\alpha - 1)e^{a(2\alpha\Delta+1)}(e^{2a\alpha}-1) \geq \alpha e^{a\Delta}(e^{2a\alpha}-e^{2a}) \\
&\iff e^{a(2\alpha - 1) \Delta} \cdot \frac{\sinh(a\alpha)}{a\alpha} \geq \frac{\sinh((\alpha-1)a)}{(\alpha-1)a},
\end{align*}
which is true because $a(2\alpha - 1) \Delta > 0$ and $\frac{\sinh x}{x}$ is increasing (\Cref{lem:sinhx-divx-inc}).

Altogether, the above argument shows that $f(\Delta) \ge 0$ for all $\Delta \geq 1$. In other words, the $\eps$-DP discrete Laplace mechanism is $a (\Delta + \left(e^{-a \Delta }-1\right) \csch(a))$-zCDP

Tightness of the bound follows directly from the tightness in \Cref{lem:dlap-rdp} by taking $\alpha \to 1^+$.
\end{proof}

\section{Local DP Mechanisms}

We next move on towards local DP mechanisms~\cite{KasiviswanathanLNRS11} where the randomizer is now run on a single data point (i.e. $n = 1$ in \Cref{def:dp,def:rdp}). We consider two popular local DP mechanisms: $k$-Randomized Response and RAPPOR.

\section{$k$-Randomized Response}

The $k$-Randomized Response ($k$-RR) mechanism~\cite{KairouzBR16} takes as input a number in $[k] := \{1, \dots, k\}$ outputs its input with probability $p = \frac{e^\eps}{e^\eps + k -1}$, and a uniformly random other output symbol with probability $1-p$. (Note that under the notions of \Cref{def:dp,def:rdp}, both the domain $\cX$ and the range $\cY$ are $[k]$.)

We prove that the tight zCDP bound is obtained at $\alpha \to 1$ for $2 \leq k \leq 6$:

\begin{theorem}\label{thm:zcdp-rr-tight}
The $\eps$-DP $k$-RR mechanism is tightly-$\left(\eps \cdot \frac{e^\eps - 1}{e^\eps + k - 1}\right)$-zCDP for $2 \le k \le 6$.
\end{theorem}

Interestingly, the condition on $k$ is necessary, as we show below that for any $\eps > 0$ and any sufficiently large $k > k^*(\eps)$, the tight bound is \emph{not} obtained at $\alpha \to 1$. Furthermore, the bound $k^*(\eps)$ below satisfies $\lim_{\eps \to 1} k^*(\eps) = 6$; this means that 6 in \Cref{thm:zcdp-rr-tight} is the best possible value one can hope for.

\begin{lemma} \label{lem:non-opt-rr}
For all $\eps > 0$ and $k > k^*(\eps) := \frac{2 \left(e^{\eps }-1\right) \left(e^{\eps }-1-\eps\right)}{\eps(e^{\eps } + 1) -2 e^{\eps }+2}$, the $\eps$-DP $k$-RR mechanism is \emph{not} $\left(\eps \cdot \frac{e^\eps - 1}{e^\eps + k - 1}\right)$-zCDP.
\end{lemma}

Finally, we also obtain zCDP bounds for the case $k > 6$. First, we note that the RDP curve (see \Cref{prop:rdp-k-rr}) is clearly decreasing in $k$, the zCDP bound must too. Thus, we immediately obtain the following as a corollary of \Cref{thm:zcdp-rr-tight}. 

\begin{corollary}\label{cor:krr-general}
    The $\eps$-DP $k$-RR mechanism is $\left(\eps \cdot \frac{e^\eps - 1}{e^\eps - 1 + \min(k, 6)}\right)$-zCDP.
\end{corollary}

An undesirable aspect of the above corollary is that the bound does not converge to 0 as $k \to \infty$. As such, we derive a bound with such behavior below. (Note that as $k \to \infty$, the zCDP bound below becomes $O\left(\frac{\eps^2}{\log k}\right)$.)

\begin{theorem} \label{thm:rr-loose-large-k}
The $\eps$-DP $k$-RR mechanism is $\eps^2 \cdot \max\left\{\frac{1}{\log\left(\frac{1}{\eps}\sqrt{k - 1 + e^{\eps}}\right)}, \frac{1}{\sqrt{k - 1 + e^{\eps}}}\right\}$-zCDP.
\end{theorem}

We stress that the bounds in \Cref{cor:krr-general} and \Cref{thm:rr-loose-large-k} are \emph{not} tight and it remains an interesting question to obtain tight bounds for these regimes of parameters. 

\subsection{RDP Bound Derivation}

We start by computing the tight bounds for RDP guarantees of $k$-RR.

\begin{proposition}\label{prop:rdp-k-rr}
The $\eps$-DP $k$-RR mechanism is tightly-$(\alpha, \hat{\eps}(\alpha))$-RDP for all $\alpha > 1$ where
\begin{align*}
\heps_{RR}(\alpha) &=
\frac{1}{\alpha - 1} \log\left(
\frac{e^{\alpha \eps} + e^{(1-\alpha)\eps} + k - 2}
{k - 1 + e^\eps}
\right).
\end{align*}
Furthermore, we have $\hat{\eps}_{RR}(1) := \lim_{\alpha \to 1} \heps_{RR}(\alpha)$ is equal to $\eps \cdot \frac{e^\eps - 1}{e^\eps - 1 + k}$.
\end{proposition}
\begin{proof}
 Without loss of generality (due to symmetry), it suffices to consider distributions $P$ and $Q$, which are the distribution of the mechanism's output on the first symbol and second symbol, respectively. Let $p = \frac{e^\eps}{k-1+e^\eps}$ and $q = \frac{1}{k-1+e^\eps}$, then
\begin{align*}
e^{(\alpha - 1) \hat{\eps}_{RR}(\alpha)} = e^{(\alpha - 1)\dr{P}{Q}}= \sum_{x \in X} f_P(x)^\alpha f_Q(x)^{1-\alpha} &= p^\alpha q^{1-\alpha} + q^\alpha p^{1-\alpha} + (k-2)q\\
&= \frac{e^{\alpha \eps} + e^{(1-\alpha)\eps} + k - 2}{k - 1 + e^\eps}.
\end{align*}

$\hat{\eps}_{RR}(1)$ can be computed by a simple application of L'Hôpital's Rule.
\end{proof}


\subsection{From RDP to zCDP}

A tight CDP bound for $k$-RR is difficult in general, because the RDP curve is not always concave in $\alpha$. We derive tight bounds in the cases where the curve is concave. 

\begin{proof}[Proof of \Cref{thm:zcdp-rr-tight}]
We will show the following two properties of $\heps_{RR}$ (assuming that $2 \leq k \leq 6$)\footnote{It is straightforward to improve upon \Cref{thm:zcdp-rr-tight} and \Cref{cor:krr-general} with the current proof technique by taking more complex constraints on $k$. We omit this for sake of presentation, but it allows extending the tightness bound for moderate $k > 6$.}:
\begin{enumerate}[(A)]
\item $\lim_{\alpha \to 1^+} \heps_{RR}(\alpha) = \heps(1)$, and, \label{enum:increasing}
\item $\heps''(\alpha) \leq 0$ for all $\alpha > 1$. (i.e. concavity) \label{enum:concave}
\end{enumerate}
From \Cref{lem:concavity}, these two properties imply that $\frac{\heps_{RR}(\alpha)}{\alpha} \leq \heps_{RR}(1)$ for all $\alpha > 1$. In other words, the mechanism is $\heps_{RR}(1)$-zCDP. The tightness also follow immediately from the tightness of the RDP curve (and taking $\alpha \to 1$).

Thus, we are left to prove \Cref{enum:increasing,enum:concave}.

\paragraph{Proof of \Cref{enum:increasing}.}
Define $H(\alpha) = \log(e^{(1-\alpha)\eps}+e^{\alpha \eps} + k - 2)$. The RDP curve for $k$-RR can be written in terms of $H$ as follows:
$$
\heps_{RR}(\alpha) = \frac{H(\alpha) - H(1)}{\alpha - 1}.
$$
The first derivative of $\heps_{RR}$ is
\begin{align*}
\hat{\eps}'(\alpha) &= \frac{(\alpha -1) H'(\alpha )-H(\alpha )+H(1)}{(\alpha -1)^2}.\\
\end{align*}
Applying L'Hôpital's Rule, we get
\begin{align*}
\lim_{\alpha \to 1} \hat{\eps}'_{RR}(\alpha) = \lim_{\alpha \to 1} \frac{H''(\alpha)}{2} = \frac{\eps^2(k + e^\eps - 1)(e^\eps + 1) - \eps^2(e^\eps - 1)^2}{2\left(k+e^{\eps }-1\right)^2} = \frac{\eps ^2 \left((k+2) e^{\eps }+k-2\right)}{2 \left(k+e^{\eps }-1\right)^2}.
\end{align*}

Thus,
\begin{align}
\lim_{\alpha \to 1^+} \hat{\eps}'_{RR}(\alpha) \le 
\hat{\eps}_{RR}(1) &\iff
\frac{\eps ^2 \left((k+2) e^{\eps }+k-2\right)}{2 \left(k+e^{\eps }-1\right)^2}
\le \eps \frac{e^\eps - 1}{e^\eps - 1 + k} \nonumber \\
&\iff k\le \frac{2 \left(e^{\eps }-1\right) \left(e^{\eps }-1-\eps\right)}{\eps(e^{\eps } + 1) -2 e^{\eps }+2} \label{eq:cond-inc-rr} 
\end{align}
We conclude by noting this constraint is satisfied when $2 \le k \le 6$, i.e. the RHS of \eqref{eq:cond-inc-rr} is at least 6.

To see this, consider the Taylor's expansion of 6 times the denominator, we have
\begin{align*}
6\left(\eps(e^{\eps } + 1) -2 e^{\eps }+2\right)
&= 6\left((e^{\eps} - 1 - \eps) \eps - 2\left(e^\eps - 1 - \eps - \frac{\eps^2}{2}\right)\right) \\
&= 6\left(\left(\sum_{i=2}^\infty \frac{\eps^i}{i!}\right) \cdot \eps - 2\left(\sum_{i=3}^\infty \frac{\eps^i}{i!}\right)\right) \\
&= \left(\eps^2\right) \cdot \left(\sum_{i=1}^{\infty}  \eps^i \cdot 6\left(\frac{1}{(i + 1)!} - \frac{2}{(i + 2)!}\right)\right) \\
&= \left(\eps^2\right) \cdot \left(\sum_{i=1}^{\infty}  \frac{\eps^i}{i!} \cdot \frac{6i}{(i + 1)(i + 2)}\right) \\
&\leq \left(2 \sum_{i=2}^{\infty} \frac{\eps^i}{i!}\right)\left(\sum_{i=1}^{\infty}  \frac{\eps^i}{i!}\right)\\
&= 2 \left(e^{\eps }-1-\eps\right) \left(e^{\eps }-1\right),
\end{align*}
where the inequality is true term-by-term. Thus, \eqref{eq:cond-inc-rr} always holds for $k \leq 6.$

\paragraph{Proof of \Cref{enum:concave}.}
%
The second derivative of $\heps_{RR}$ is
\begin{align*}
\heps''_{RR}(\alpha) = \frac{(\alpha -1)^2 H''(\alpha )-2 (\alpha -1) H'(\alpha )+2 H(\alpha )-2 H(1)}{(\alpha -1)^3}.
\end{align*}

Let $\psi(\alpha)$ denote the numerator of $\heps''_{RR}(\alpha)$ above. To show that $\heps''_{RR}(\alpha) \geq 0$, it suffices to show that $\psi(\alpha) \geq 0$ for all $\alpha > 1$. Observe that $\lim_{\alpha \to 1^+} \psi(\alpha) = 0$. Thus, to prove the non-negativity of $\psi(\alpha)$, it in turn suffices to show that $\psi'(\alpha) \le 0$ for all $\alpha > 1$.

To show this, note that $\psi'(\alpha) = H'''(\alpha) (\alpha - 1)^2$, so the sign of $\psi'(\alpha)$ is determined by the sign of
$$
H'''(\alpha) = \frac{\eps ^3 \left(e^{\eps -\alpha  \eps }-e^{\alpha  \eps }\right) \left((k-2) e^{\alpha  \eps }+(k-2) e^{\eps -\alpha  \eps }-(k-2)^2+8 e^{\eps }\right)}{\left(e^{\alpha  \eps }+e^{\eps -\alpha  \eps }+k-2\right)^3},
$$

which is non-positive if the last term in the numerator is non-negative (since all other terms but $\left(e^{\eps -\alpha  \eps }-e^{\alpha  \eps }\right)$ are non-negative). It is clear that this last term is increasing in $\eps \in [0, \infty)$. This means that 
\begin{align*}
(k-2) e^{\alpha  \eps }+(k-2) e^{\eps -\alpha  \eps }-(k-2)^2+8 e^{\eps } \geq (k - 2) + (k - 2) - (k - 2)^2 + 8 = k(6 - k),
\end{align*}
which is non-negative for $k \leq 6$. Thus, $H'''(\alpha) \leq 0$ and we can conclude that $\heps''_{RR}(\alpha)$ is concave.
\end{proof}

\subsection{Non-Optimality of $\alpha = 1$}

The above proof also yields a rather simple criterion to certify that $\alpha = 1$ is not the optimal for sufficiently large $k$, as formalized below.

\begin{proof}[Proof of \Cref{lem:non-opt-rr}]
From \Cref{eq:cond-inc-rr}, in this regime we have $\lim_{\alpha \to 1^+} \hat{\eps}'_{RR}(\alpha) > \heps_{RR}(1)$. Since $\frac{\heps'_{RR}(\alpha)}{2\alpha - 1}$ is continuous (in $\alpha$), there exists $\alpha_0 > 1$ such that $\frac{\heps'_{RR}(\alpha)}{2\alpha - 1} > \heps_{RR}(1)$ for all $\alpha \in (1, \alpha)$. Thus, we have
\begin{align*}
\heps_{RR}(\alpha_0) - \alpha_0 \cdot \heps_{RR}(1) = \int_{1}^{\alpha_0} (\heps'_{RR}(\alpha) - \heps_{RR}(1))  \mathrm{d} \alpha > 0.
\end{align*}
Thus, $\frac{\heps_{RR}(\alpha_0)}{\alpha_0} > \heps_{RR}(1)$, implying that the mechanism is not $\heps_{RR}(1)$-zCDP.
\end{proof}

\subsection{Loose zCDP for Large $k$ Regime}

Finally, we prove a loose bound that converges to zero as $k \to \infty$ via a commonly used approach of separately considering smaller and larger $\alpha$.

\begin{proof}[Proof of \Cref{thm:rr-loose-large-k}]
Let $\rho = \eps^2 \cdot \max\left\{\frac{1}{\log\left(\eps\sqrt{k - 1 + e^{\eps}}\right)}, \frac{1}{\sqrt{k - 1 + e^{\eps}}}\right\}$. Let $\alpha^* = \frac{\eps}{\rho}$. Consider two cases based on the value of $\alpha$:
\begin{itemize}
\item Case I: $\alpha \geq \alpha^*$. In this case, we simply have $\frac{\heps_{RR}(\alpha)}{\alpha} \leq \frac{\eps}{\alpha} \leq \rho$, where the first inequality follows from the fact that any $\eps$-DP algorithm is $(\alpha, \eps)$-DP for all $\alpha > 1$.
\item Case II: $\alpha < \alpha^*$. In this case, we have
\begin{align*}
\frac{\heps(\alpha)}{\alpha} \leq \heps(\alpha) &= \frac{1}{\alpha - 1} \cdot \log \left(\frac{e^{\alpha \eps} + e^{(1-\alpha)\eps} + k - 2}{k - 1 + e^\eps}\right) \\
&= \frac{1}{\alpha - 1} \cdot \left(\frac{e^{\alpha \eps} + e^{(1-\alpha)\eps} + k - 2}{k - 1 + e^\eps} - 1\right) \\
&= \frac{1 - e^{(1 - \alpha)\eps}}{\alpha - 1} \cdot \frac{e^{\alpha \eps} - 1}{k - 1 + e^\eps} \\
&\leq \eps \cdot \frac{e^{\alpha^* \eps}}{k -  1 + e^\eps} \\
&= \eps \cdot \frac{e^{\eps^2 / \rho}}{k -  1 + e^\eps} \\
&\leq \frac{\eps^2}{\sqrt{k - 1 + e^{\eps}}}\\
&\leq \rho.
\end{align*}
\end{itemize}
Thus, the mechanism is $\rho$-zCDP as claimed.
\end{proof}
\subsection{RAPPOR}

Next, we consider the (basic) RAPPOR mechanism~\cite{ErlingssonPK14}.
The mechanism takes input from $[d]$ and outputs a $d$-bit string (i.e. $\cX = [d], \cY = \{0, 1\}^d$ in \Cref{def:dp,def:rdp}). It first encodes its input as a one-hot vector, and then flips each bit in the vector independently. To satisfy $\eps$-DP, its output distribution is
$$
Pr[M(X) = x] = \prod_{i=1}^d
    \begin{cases}
    \frac{e^{\eps/2}}{e^{\eps/2} + 1} & \text{ if $\hot(X)_i = x_i$}\\
    \frac{1}{e^{\eps/2} + 1} & \text{ if $\hot(X)_i \neq x_i$}\\
    \end{cases},
$$
where $\hot(X)$ refers to the $d$-length one-hot vector where the $X$th entry is 1 and the rest are 0.

The tight RDP and zCDP bounds for RAPPOR follows almost immediately from those of Binary RR (from the previous section), as formalized below.

\begin{proposition}
The $\eps$-DP RAPPOR mechanism is tightly-$(\alpha, \hat{\eps}_{RAP}(\alpha))$-RDP for all $\alpha > 1$ where 
$$
\hat{\eps}_{RAP}(\alpha) = \frac{2}{\alpha - 1}\log\left(
\frac{e^{\alpha \eps/2} + e^{(1-\alpha) \eps/2}}
{e^{\eps/2} + 1}
\right).
$$
\end{proposition}
\begin{proof}
Without loss of generality (due to symmetry), it suffices to consider distributions $P$ and $Q$, which are the distribution of the mechanism's output on the first symbol and second symbol, respectively. Let $P_i, Q_i$ denote the distributions of the $i$-th coordinate of $P, Q$ respectively. Notice that $P, Q$ are product distributions, and that $P_i = Q_i$ for all $i \notin \{1, 2\}$. Thus, from additivity of R{\'e}nyi divergence, we have
\begin{align*}
\hat{\eps}_{RAP}(\alpha) = \dr{P}{Q}
= \dr{P_1}{Q_1} + \dr{P_2}{Q_2}.
\end{align*}
Finally, observe that $P_i, Q_i$ for $i \in \{1, 2\}$ are exactly the same distribution as that of the $(\eps/2)$-DP Binary RR. Thus, we may invoke the bound from \Cref{prop:rdp-k-rr} to arrive at the claimed bound. 
\end{proof}

\begin{proposition}\label{prop:zcdp-rappor}
The $\eps$-DP RAPPOR mechanism is tightly-$\left(\eps \cdot \tanh(\eps / 4)\right)$-zCDP.
\end{proposition}
\begin{proof}
    This follows immediately from \Cref{thm:zcdp-rr-tight} since the RDP bound is exactly twice the RDP bound for $(\eps / 2)$-DP Binary RR.
\end{proof}
\section{$\eta$-Bounded Range Mechanisms} 

Finally, we consider $\eta$-Bounded Range ($\eta$-BR) mechanisms , and we prove tight zCDP bound in this case as well.

\begin{theorem}\label{thm:br-cdp}
Any $\eta$-BR mechanism satisfies $\rho$-CDP where
$$
\rho =\frac{\eta}{e^\eta - 1} + \log\left(
\frac{e^{\eta }-1}{\eta }\right) - 1.
$$
Furthermore, this is tight\footnote{Note that we do \emph{not} use the terminology tightly-$(\alpha, \rho)$-zCDP directly here, since ``$\eta$-BR mechanisms'' constitute a class of mechanisms rather than a single mechanism.}, i.e. there exists an $\eta$-BR mechanism which is tightly-$\rho(\eta)$-CDP.
\end{theorem}

\subsection{RDP Bound Derivation}

We start by proving the following lemma, which yields a bound for the R{\'e}nyi divergence when given a parameter $t$ in \Cref{def:br}.

\begin{lemma}\label{lem:br-rdp-one}
Let $t, \eta$ be real numbers such that $\eta > 0$ and $0 \leq t \leq \eta$, and
let $P$ and $Q$ be distributions such that $\log\left(\frac{P(x)}{Q(x)}\right) \in [-t, -t + \eta]$ for all $x \in \supp(P) \cup \supp(Q)$. Then, for all $\alpha > 1$, we have
$$
\dr{P}{Q} \le \frac{1}{\alpha-1} \log\left(\frac{e^\eta - e^t}{e^\eta - 1} e^{-t \alpha} +
\frac{e^t - 1}{e^\eta - 1} e^{(-t + \eta)\alpha}
\right)
=: \heps_{BR}(t; \alpha).
$$
Moreover, there exist a pair of distributions $P, Q$ with $\log\left(\frac{P(x)}{Q(x)}\right) \in \{-t, -t + \eta\}$ for all $x \in \supp(P) \cup \supp(Q)$ such that the above inequality is an equality.
\end{lemma}
\begin{proof}
We will first show achieveability of the bound by describing a pair of distributions $P$ and $Q$. Define $P$ and $Q$ as follows:
\begin{align*}
f_P(x) &= \begin{cases}
\frac{e^{-t} \left(e^{\eta }-e^t\right)}{e^{\eta }-1} & \text{ if $x = 0$}\\
\frac{\left(e^t-1\right) e^{\eta -t}}{e^{\eta }-1} & \text{ if $x = 1$}\\
\end{cases}\\
f_Q(x) &= \begin{cases}
\frac{e^{\eta }-e^t}{e^{\eta }-1} & \text{ if $x = 0$}\\
\frac{e^t-1}{e^{\eta }-1} & \text{ if $x = 1$}\\
\end{cases}
\end{align*}

It is easy to show that $P$ and $Q$ are valid probability distributions, $\frac{f_P(0)}{f_Q(0)} = e^{-t}$ and $\frac{f_P(1)}{f_Q(1)} = e^{-t + \eta}$ as desired, and that the R{\'e}nyi divergence is exactly equal to the required bound.

We next prove that this is the worst case bound, using a proof technique due to \cite{bun-steinke-16, DPorg-pdp-to-zcdp}.
Define the randomized rounding function $A: [e^{-t}, e^{-t + \eta}] \to \{e^{-t}, e^{-t + \eta}\}$ such that 
\begin{align*}
\mathbb{P}(A(z) = e^{-t}) &= \frac{e^\eta - z e^t}{e^\eta - 1}\\
\mathbb{P}(A(z) = e^{-t + \eta}) &= \frac{z e^t - 1}{e^\eta - 1}
\end{align*}
Note that this satisfies $\mathbb{E}[A(z)]=z$. Thus, by Jensen's inequality (and by the convexity of $x \mapsto x^\alpha$), 

$$
z^\alpha = \mathbb{E}[A(z)]^\alpha \le \mathbb{E}[A(z)^\alpha] =
\frac{e^\eta - z e^t}{e^\eta - 1} e^{-t \alpha} +
\frac{z e^t - 1}{e^\eta - 1} e^{(-t + \eta)\alpha}.
$$

Consider any distributions $P$ and $Q$ such that $\log\left(\frac{P(x)}{Q(x)}\right) \in [-t, -t + \eta]$ for all $x \in \supp(P) \cup \supp(Q)$.
Setting $z = P(x) / Q(x)$, we have
\begin{align*}
e^{(\alpha - 1)\dr{P}{Q}} &= \mathbb{E}_{x \sim Q}\left[\left(
\frac{P(x)}{Q(x)}
\right)^\alpha\right]\\
&\le \mathbb{E}_{x \sim Q}\left[\frac{e^\eta - \frac{P(x)}{Q(x)} e^t}{e^\eta - 1} e^{-t \alpha} +
\frac{\frac{P(x)}{Q(x)} e^t - 1}{e^\eta - 1} e^{(-t + \eta)\alpha}
\right]\\
&= \mathbb{E}_{x \sim Q}\left[\frac{e^\eta - \frac{P(x)}{Q(x)} e^t}{e^\eta - 1} e^{-t \alpha}\right] +
\mathbb{E}_{x\sim Q}\left[\frac{\frac{P(x)}{Q(x)} e^t - 1}{e^\eta - 1} e^{(-t + \eta)\alpha}
\right]\\
&= \frac{e^\eta - e^t}{e^\eta - 1} e^{-t \alpha} +
\frac{e^t - 1}{e^\eta - 1} e^{(-t + \eta)\alpha}.&\qedhere
\end{align*}
\end{proof}

The above lemma allows us to easily compute the tight RDP bound for $\eta$-BR mechanism by optimizing over the value of $t$.

\begin{theorem}[RDP for Bounded Range Mechanisms] \label{thm:br-rdp}
Any $\eta$-BR mechanism satisfies $(\alpha, \hat{\eps}_{BR}(\alpha))$-RDP for all $\alpha > 1$ where 
\begin{align*}
\hat{\eps}_{BR}(\alpha) &= \frac{1}{\alpha-1} \log\left(\frac{(e^{\alpha \eta} - 1)^\alpha \left(\frac{\alpha (e^{\alpha \eta} - e^{\eta})}{\alpha - 1}\right)^{1-\alpha}} {\alpha (e^\eta - 1)}\right).
\end{align*}
Furthermore, the above bound is tight, i.e. there exists an $\eta$-BR mechanism which is tightly-$(\alpha, \hat{\eps}_{BR}(\alpha))$-RDP. Moreover, we have $\hat{\eps}_{BR}(1) := \lim_{\alpha \to 1} \heps_{BR}(\alpha)$ is equal to $\frac{\eta}{e^\eta - 1} + \log\left(\frac{e^\eta - 1}{\eta}\right) - 1$.
\end{theorem}
\begin{proof}
From \Cref{lem:br-rdp-one}, any $\eta$-BR mechanism is $(\alpha, \heps(\alpha))$-RDP where $\heps_{BR}(\alpha) = \sup_{t \in [0, \eta]} \heps_{BR}(t; \alpha)$. Since $\heps_{BR}(t; \alpha)$ is continuous and differentiable, it must achieve a maximum value at some $t = t^* \in [0, \eta]$, and $t^*$ must satisfy (i) $t^* \in \{0, \eta\}$, or (ii) its derivative is zero. It is clear that we cannot be in case (i) since $\heps_{BR}(0; \alpha) = \heps_{BR}(\eta; \alpha) = 0$ but $\heps_{BR}(t; \alpha) > 0$ for $t \in (0, \eta)$.

The derivative of $\heps_{BR}(t; \alpha)$ is
\begin{align*}
\frac{\mathrm{d}}{\mathrm{d}t} \heps_{BR}(t; \alpha)
&= 
\frac{1}{\alpha - 1} \left(
\frac{e^{\alpha  \eta +t}-e^t}{e^{\eta }+\left(e^t-1\right) e^{\alpha  \eta }-e^t}-\alpha\right).
\end{align*}
Its only root (in $[0, \eta]$) is $t^* = \log \left(\frac{\alpha  \left(e^{\alpha  \eta }-e^{\eta }\right)}{(\alpha -1) \left(e^{\alpha  \eta }-1\right)}\right)$. Thus, we have
\begin{align*}
\heps_{BR}(\alpha) = \heps_{BR}(t^*; \alpha) =  \frac{1}{\alpha-1} \log\left(\frac{(e^{\alpha \eta} - 1)^\alpha \left(\frac{\alpha (e^{\alpha \eta} - e^{\eta})}{\alpha - 1}\right)^{1-\alpha}} {\alpha (e^\eta - 1)}\right).
\end{align*}

The tightness follows from that of \Cref{lem:br-rdp-one} by letting $M: \{0, 1\} \to \supp(P) \cup \supp(Q)$ be the mechanism such that $M(0) \sim P$ and $M(1) \sim Q$ where $P, Q$ are the tight example for $t = t^*$ from \Cref{lem:br-rdp-one}.

 It remains to compute
\begin{align*}
\hat{\eps}_{BR}(1) = \lim_{\alpha \to 1} \heps_{BR}(\alpha) = \lim_{\alpha \to 1}
\frac{1}{\alpha-1} \log\left(
\frac{(e^{\alpha \eta} - 1)^\alpha 
\left(\frac{\alpha (e^{\alpha \eta} - e^{\eta})}{\alpha - 1}\right)^{1-\alpha}
}
{\alpha (e^\eta - 1)}
\right).
\end{align*}

This is an indeterminate form 0/0, so we can apply L'Hôpital's Rule by considering $f(\alpha) = \frac{N(\alpha)}{D(\alpha)}$ with $D(\alpha) = \alpha - 1$. We have $D'(\alpha) = 1$, and

\begin{align*}
N'(\alpha) &= 
\frac{(\alpha -1) e^{\eta } \eta }{e^{\eta }-e^{\alpha  \eta }}
+\frac{\alpha  \eta }{e^{\alpha  \eta }-1}
-\log \left(\frac{\alpha  \left(e^{\alpha  \eta }-e^{\eta }\right)}{(\alpha -1) \left(e^{\alpha  \eta }-1\right)}\right)
+\eta\\
\lim_{\alpha \to 1} N'(\alpha) &= 
-1 + \frac{\eta}{e^\eta - 1} + \eta - \log\left(
\lim_{\alpha \to 1} \frac{\alpha  \left(e^{\alpha  \eta }-e^{\eta }\right)}{(\alpha -1) \left(e^{\alpha  \eta }-1\right)} \right)\\
&= -1 + \frac{\eta}{e^\eta - 1} + \eta - \log\left(
\frac{e^{\eta } \eta }{e^{\eta }-1}
\right)\\
&= \frac{\eta}{e^\eta - 1} + \log\left(
\frac{e^{\eta }-1}{\eta }\right) - 1.&\qedhere
\end{align*}
\end{proof}

\subsection{From RDP to zCDP}

Finally, we show below that any $\eta$-BR mechanism is $\rho$-zCDP for $\rho = \heps_{BR}(1)$ (\Cref{thm:br-cdp}). 

\begin{proof}[Proof of \Cref{thm:br-cdp}]
Let $f(\alpha) := \frac{\hat{\eps}_{BR}(\alpha)}{\alpha}$ where $\hat{\eps}_{BR}(\alpha)$ is from \Cref{thm:br-rdp}.
Our proof will proceed by showing that $f$ is decreasing in $\alpha$ (for $\alpha > 1$), which implies that any $\eta$-BR mechanism is $\heps_{BR}(1)$-zCDP. To do this, let
\begin{align*}
g:= &\frac{\mathrm{d}}{\mathrm{d}\alpha} f(\alpha) = \\&= \frac{(1-2 \alpha ) \log \left(e^{\alpha  \eta }-e^{\eta }\right)+(2 \alpha -1) \log \left((\alpha -1) \left(e^{\eta }-1\right)\right)+\alpha  \left((\alpha -1) \eta +\alpha  \log \left(\frac{\alpha  \left(e^{\alpha  \eta }-e^{\eta }\right)}{(\alpha -1) \left(e^{\alpha  \eta }-1\right)}\right)\right)}{(\alpha -1)^2 \alpha ^2}\\
&+\frac{\eta }{(\alpha -1) \left(e^{\alpha  \eta }-1\right)} +\frac{e^{\eta } \eta }{\alpha  e^{\eta }-\alpha  e^{\alpha  \eta }}
\end{align*}

Let us view $g$ as function of $\eta$, i.e. $g = g_\alpha(\eta)$. Recall that we wish to show that $g_\alpha(\eta) \le 0$ for all $\eta > 0, \alpha > 1$. Simple calculation verifies that $\lim_{\eta \to 0^+}g_\alpha(\eta) = 0$ for all $\alpha > 1$, so the desired inequality will follow from showing that $g'_\alpha(\eta) \le 0$. For this, we first compute the derivative:
\begin{align*}
g'_\alpha(\eta) &= \frac{
h_1(\alpha, \eta)  h_2(\alpha, \eta)
\text{csch}^2\left(\frac{\alpha  \eta }{2}\right) \text{csch}^2\left(\frac{1}{2} (\eta -\alpha  \eta )\right)}{8 (\alpha -1)^2 \alpha ^2}
\end{align*}
where
\begin{align*}
h_1(\alpha) &= \left(-\alpha ^2 \cosh (\eta -\alpha  \eta )+(\alpha -1)^2 \cosh (\alpha  \eta )+2 \alpha -1\right)\\
h_2(\alpha) &= (\alpha -1) \alpha  \eta +\sinh (\alpha  \eta )-\coth \left(\frac{\eta }{2}\right) (\cosh (\alpha  \eta )-1).
\end{align*}

Given the rest of the terms are trivially non-negative, it remains to show that $h_1$ is positive and $h_2$ is negative. To do this, recall the identities $\cosh(2x) = 1 + 2\sinh^2 x$ and $\sinh(x - y) = \sinh x \cosh y - \cosh x \sinh y$.

Using the first identity, we can simplify $h_1 \ge 0$ as follows
\begin{align*}
h_1(\alpha) \ge 0 &\iff (\alpha -1)^2 \cosh (\alpha  \eta ) + \alpha ^2-(\alpha -1)^2 \ge \alpha ^2 \cosh (\eta (\alpha - 1))\\
&\iff \frac{\cosh (\alpha  \eta )-1}{\alpha^2} \ge \frac{\cosh (\eta (\alpha - 1))-1}{(\alpha - 1)^2}\\
&\iff \frac{\sinh^2(\alpha  \eta / 2)}{\alpha^2} \ge \frac{\sinh^2(\eta (\alpha - 1)/2)}{(\alpha - 1)^2}\\
&\iff \frac{\sinh(\alpha  \eta / 2)}{\alpha  \eta / 2} \geq \frac{\sinh(\eta (\alpha - 1)/2)}{\eta(\alpha - 1)/2},
\end{align*}
and the last inequality is true due to \Cref{lem:sinhx-divx-inc}.

Similarly, using the two identities, we can simplify $h_2$ as follows:
\begin{align*}
(\alpha -1) \alpha  \eta - h_2 &= -2\sinh \left(\frac{\alpha  \eta}{2}\right)\cosh \left(\frac{\alpha  \eta}{2}\right)+\coth \left(\frac{\eta }{2}\right) \left(2\sinh^2 \left(\frac{\alpha  \eta}{2}\right)\right) \\
&= 2\sinh \left(\frac{\alpha  \eta}{2}\right) \left(\frac{\cosh \left(\frac{\eta }{2}\right)\sinh \left(\frac{\alpha  \eta}{2}\right) - \sinh\left(\frac{\eta }{2}\right) \cosh \left(\frac{\alpha  \eta}{2}\right)}{\sinh\left(\frac{\eta }{2}\right)}\right) \\
&= 2\sinh \left(\frac{\alpha  \eta}{2}\right) \left(\frac{\sinh \left(\frac{(\alpha  - 1)\eta }{2}\right)}{\sinh\left(\frac{\eta}{2}\right)}\right) \\
&= \alpha \cdot 2\sinh \left(\frac{(\alpha - 1)\eta}{2}\right) \cdot \left(\frac{\frac{\sinh \left(\frac{\alpha\eta }{2}\right)}{\frac{\alpha\eta }{2}}}{\frac{\sinh\left(\frac{\eta}{2}\right)}{\frac{\eta}{2}}}\right) \\
(\text{\Cref{lem:sinhx-divx-inc}}) &\geq \alpha \cdot 2\left(\frac{(\alpha - 1)\eta}{2}\right) \cdot 1 \\ &= \alpha(\alpha - 1) \eta.
\end{align*}
Thus, $h_2 \leq 0$ as desired.

This concludes the proof that $f$ is decreasing in $\alpha$. Hence, any $\eta$-BR mechanism is $\heps_{BR}(1)$-zCDP. 

The tightness follows from letting $M: \{0, 1\} \to \{0, 1\}$ be the mechanism\footnote{Note that we define a new mechanism here rather than using the tightness in \Cref{thm:br-rdp} and taking $\alpha \to 1$ because the tight mechanisms in \Cref{thm:br-rdp} are different for different values of $\alpha$.} such that
\begin{align*}
\bbP(M(0) = y) &= \begin{cases}
\frac{e^{-t} \left(e^{\eta }-e^t\right)}{e^{\eta }-1} & \text{ if $y = 0$}\\
\frac{\left(e^t-1\right) e^{\eta -t}}{e^{\eta }-1} & \text{ if $y = 1$}\\
\end{cases}\\
\bbP(M(1) = 1) &= \begin{cases}
\frac{e^{\eta }-e^t}{e^{\eta }-1} & \text{ if $x = 0$}\\
\frac{e^t-1}{e^{\eta }-1} & \text{ if $x = 1$}\\
\end{cases}
\end{align*}
for $t = \eta - \log\left(\frac{e^{\eta} - 1}{\eta}\right)$.
It is simple to verify that $M$ is $\eta$-BR and that $\dra{M(0)}{M(1)}{1} = \heps_{BR}(1)$, which implies the claimed tightness.
\end{proof}

\section{Conclusion}
In this work, we continue the direction started by \cite{DPorg-pdp-to-zcdp} in finding \emph{exactly} tight zCDP bounds for differentially private mechanisms. In particular, we derive tight zCDP bounds for fundamental mechanisms, including Laplace and Discrete Laplace mechanisms, $k$-Randomized Response (for sufficiently small $k$) and RAPPOR. Given the ubiquity of these mechanisms and the wide adoption of zCDP for privacy accounting,
we hope that our precise characterizations provide additional fundamental tools for accurate privacy accounting both in theory and practice. 

The obvious open problem from our work is to give a precise characterization for zCDP of $k$-RR when $k$ is large. It is unclear if a closed-form expression can be derived in this setting. Nevertheless, an efficient numerical algorithm for computing a precise bound should still be useful in practice. Finally, there are still many widely used mechanisms for which tight zCDP characterizations are not yet known (e.g. Sparse Vector Technique~\cite{DworkNRRV09}, continuous and discrete staircase mechanism \cite{geng2014optimal}, report noisy max / permute and flip \cite{dworkrothbook, mckenna2020permute, ding2021permute}, or the optimized unary encoding mechanism \cite{wang-opt-unary}); it would be interesting to prove tight bounds for them as well.

\begin{figure}[t]
  \centering
  \includegraphics[width =.8\linewidth]{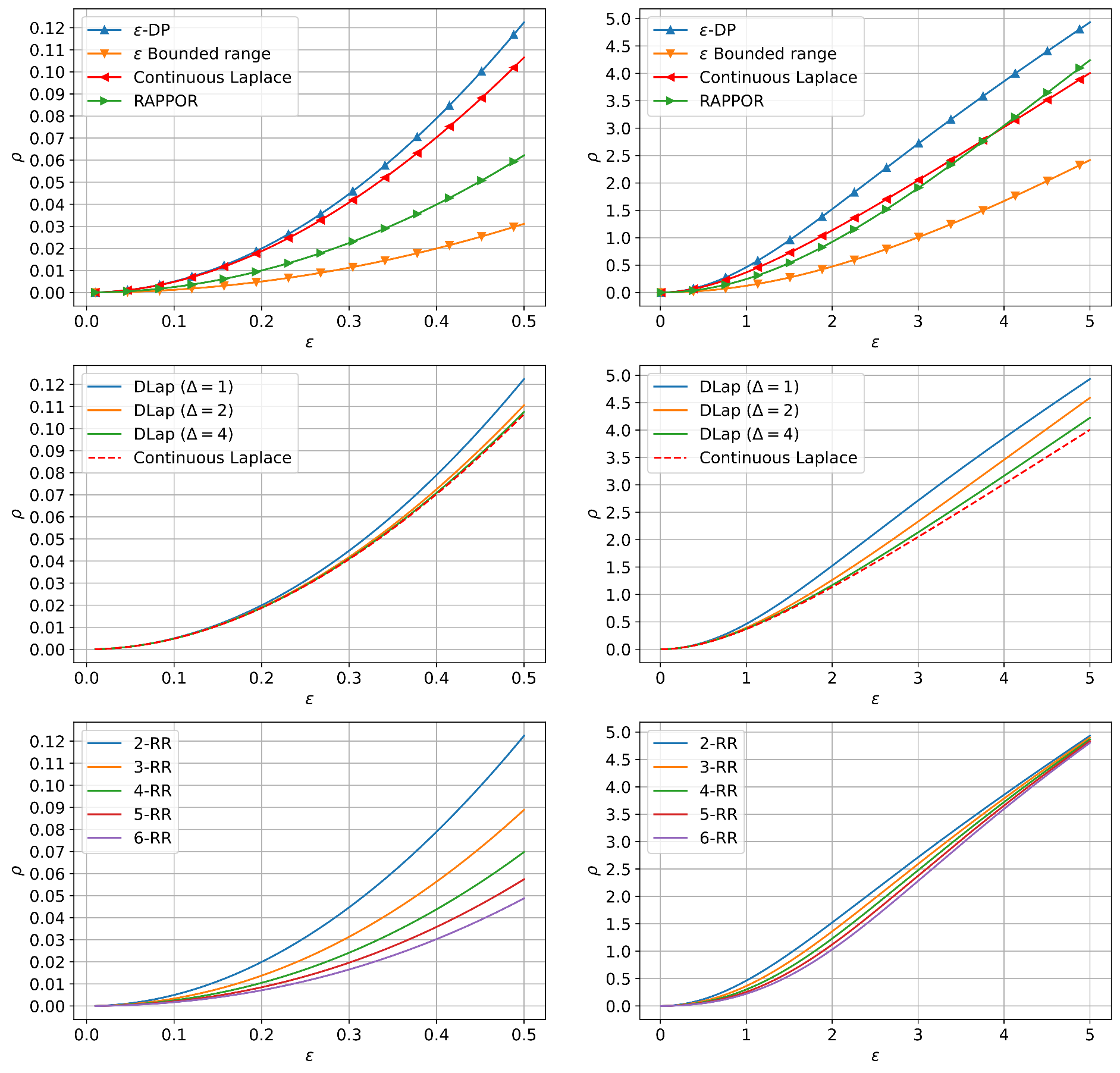}
  \caption{Tight zCDP bounds for all mechanisms as a function of $\eps$. All bounds use the formulas listed in \Cref{tab:main}.}
  \label{fig:others}
\end{figure}

\section*{Acknowledgment}
We would like to thank Thomas Steinke for insightful discussions, and for encouraging us to find a tight bound for bounded range mechanisms.

\bibliographystyle{alpha}
\bibliography{ref}

\end{document}